\documentclass[twocolumn]{revtex4-2}      
\usepackage{amsmath,mathtools,amsthm}
\usepackage{amssymb}
\usepackage[mathscr]{eucal}
\usepackage{bm,booktabs}
\usepackage{bbm}
\usepackage[shortlabels]{enumitem}
\usepackage{hyperref}
\usepackage[usenames,dvipsnames]{xcolor}

\usepackage{tikz}
\usepackage{graphicx}
\usepackage[caption=false]{subfig}

\newtheorem{theorem}{Theorem}[section]

\newtheorem{lemma}[theorem]{Lemma}
\newtheorem{proposition}[theorem]{Proposition}
\newtheorem{corollary}[theorem]{Corollary}
\newtheorem{assumption}{Assumption}

\newcommand{\pprob}{\mathbb{P}}

\newcommand{\expec}{\mathbb{E}}
\newcommand{\Exp}[1]{\expec\left[#1\right]}
\newcommand{\Expp}[1]{\expec_\pprob\left[#1\right]}

\newcommand{\bigO}[1]{O\left(#1\right)}

\newcommand{\supp}{\textup{supp}}

\def\Aut{{\textup{Aut}}}

\begin{document}
\title{Robust subgraph counting with distribution-free random graph analysis}

\author{Johan S.H. van Leeuwaarden}
\affiliation{Tilburg University}
\author{Clara Stegehuis}
\affiliation{ University of Twente}

\date{\today}

\begin{abstract} 
Subgraphs such as cliques, loops and stars form crucial connections in the topologies of real-world networks. Random graph models provide estimates for how often certain subgraphs appear, which in turn can be tested against real-world networks. These subgraph counts, however, crucially depend on the assumed degree distribution. Fitting a degree distribution to network data is challenging, in particular for scale-free networks with power-law degrees. In this paper we develop robust subgraph counts that do not depend on the entire degree distribution, but only on the mean and mean  absolute  deviation  (MAD), summary statistics that are easy to obtain for most real-world networks. By solving an optimization problem, we provide tight (the sharpest possible) bounds for the subgraph counts, for all possible subgraphs, and for all networks with degree distributions that share the same mean and MAD.  We identify the extremal random graph that attains the tight bounds as the graph with a specific three-point degree distribution. We leverage the bounds to obtain robust  scaling  laws  for how the numbers of subgraphs grow as function of  the  network  size.
The scaling laws indicate that sparse power-law networks are not the most extreme networks in terms of subgraph counts, but dense power-law networks are.  The robust bounds are also shown to hold for several real-world data sets. 
\end{abstract}

\maketitle

\section{Introduction}
The occurrence of specific subgraphs like cliques, loops and stars have been proven important for understanding and classifying complex networks, such as social, biological  and technological networks. The triangle is a much studied subgraph, because it
 describes local clustering and signals  community structure. Other subgraphs such as larger cliques
are also important for understanding network organization. Counting how often certain subgraphs appear, in partical in large-scale networks, is therefore a central topic in network science.
Indeed, subgraph counts might vary considerably across different networks~\cite{milo2002,shen-orr2002,onnela2005,tran2013} and any given network has a set of statistically significant subgraphs (also called motifs). 

Subgraphs in complex networks are broadly studied through random graphs, mathematically tractable models that can generate random samples of a graph in which nodes have i.i.d.~degrees~\cite{bianconi2005,itzkovitz2003subgraphs,hofstad2017b,stegehuis2019b,van2021optimal}. 
Random graph models take the degree distribution as input.
Conditional on the degree distribution, random graph properties such as average distance, clustering and clique formation can often be characterized and in turn be tested against measurements from real-world network data with the same degree distribution. Classical choices for degree distributions include the Poisson distribution and power-law distributions. 
The latter also often arises in analyzing real-world networks, as their degree distributions can often be approximated with a power-law distribution. 

A random variable obeys a power law if it is drawn from a probability
distribution
\begin{equation}
\pprob(h)\propto h^{-\tau}, \quad h\geq 0,
\end{equation}
with $\tau>0$ a constant known as the power law exponent. Many real-world networks were shown to be approximated by a power-law degree distribution with an exponent in the range $2<\tau<3$ \cite{jeong2000,vazquez2002,faloutsos1999}. 
Fitting a power law to real-world data, however, is statistically challenging~\cite{clauset2009,broido2018,voitalov2019scale}.
For small values, a power law is usually not a good fit. Therefore, 
one typically assumes that it only holds for values greater than some minimum $h_{\min}$. 
Alternatively, one can consider a family of distributions of the form
\begin{equation}\label{slowpow}
\pprob(h)\propto L(h) h^{-\tau}, \quad h\geq 0,
\end{equation}
where $L(h)$ is some slowly varying function, so that $L(ch)/L(h)\to 1$ for any $c > 0$ as $h\to\infty$. The function $L(h)$ can then account for deviations from the pure power law caused by smaller values of the distribution. Hence, both $h_{\min}$ and $L(h)$ are ways to deal with imperfection in data for smaller degree values. However, both $h_{\min}$ and $L(h)$ create extra parameters that make fitting the distribution to real-world data more difficult.

Larger values of the power law also present challenges. Most real-world data sets only follow a power law up to some maximal degree, which is often modeled by an exponential cutoff~\cite{newman2001a,mossa2002,ebel2002}. Real-world networks are in fact finite by definition, while a power law allows infinitely large values.
One way to link finite networks with the possibly infinite power-law values is to scale the maximum degree, called cutoff, as function of the network size. 
We will use this mathematical approach of scaling the cutoff as function of the network size also in this paper.

Network properties claimed to be universal should not be overly sensitive to the assumed degree distribution, especially when this distribution is hard to justify statistically. For  power laws as in \eqref{slowpow} for instance, the exact shape of $L(h)$ is often not crucial, while the tail exponent $\tau$ implies vastly different network properties. One reason for this is the variance of the degree distribution. 
When the number of nodes $n$ becomes large,  the variance grows to infinity for $\tau<3$, while the variance remains finite for $\tau>3$. This difference in variance growth  crucially influences network structure and subgraph formation \cite{ostilli2014,itzkovitz2003subgraphs,stegehuis2019b}.  

In this paper, we characterize subgraph counts in random graphs that only require partial information about the degree distribution. Inspired by the example of the complicated assessment of power laws, we assume that we only know the mean, range and mean absolute deviation (MAD) of the degree distribution. The MAD is an alternative to variance for measuring dispersion around the mean, and may be more appropriate in case of heavy tails. Indeed, MAD can deal with distributions that do not possess a finite variance, in particular the class of power-law distribution with $\tau\in(2,3)$, for which MAD remains finite while variance becomes infinite in the large-network limit when $n\to\infty$.

We shall identify the maximal subgraph count that can be achieved by all degree distributions that have the same mean, range and MAD. Consider a subgraph $H$ and the subgraph count  $\Expp{N_H}$ defined as the expected number of subgraphs $H$ that appear in a random graph. By constructing a maximization problem we determine the extremal degree distribution $\pprob$ that maximizes the subgraph count $\Expp{N_H}$ and will refer to the random graph with the extremal degree distribution as the extremal random graph. We solve this maximization problem
for the hidden-variable model~\cite{chung2002,boguna2004}, a random graph model that generates graphs with degrees that approximately follow some given distribution. The hidden-variable model is a widely applied model due to generality and tractability: many network properties of the model have been investigated, such as degree-degree correlations, clustering, typical distances and epidemic spreading~\cite{colomer2012,boguna2004,hofstad2017b,bianconi2006,pastor2014}. Typically, these properties depend strongly on the input distribution of the hidden variables. For example, whether or not the second moment of the input distribution diverges, strongly influences the degree-correlations, distances and behavior under epidemic processes. We therefore provide bounds on the subgraph counts that hold for all input distributions with the same mean and MAD.

For the hidden-variable model we first show that $\Expp{N_H}$ is a convex function of the hidden variables. 
We then employ a method from distributionally robust optimization for maximizing a convex function under certain constraints. 
We seek for the maximum of $\Expp{N_H}$, given the constraints on the degree distribution in terms of the mean, range and MAD. This gives rise to a semi-infinite linear program with as solution the robust subgraph bounds.


Here are the main contributions of this paper:
\begin{itemize}
    \item[(i)] We show that for all subgraphs $H$ the expected number of subgraphs $\Expp{N_H}$ is maximized for the extremal graph model with a three-point degree distribution. This maximal number of subgraphs only depends on the mean, MAD and range of the degree distribution, and therefore does not need further detailed assumptions on the degree distribution.
    In particular, we show that this extremal graph model is the same for all possible subgraphs $H$ and network sizes. 
    \item[(ii)] We derive scaling laws for $\Expp{N_H}$ when the network size $n$ grows to infinity. These scaling laws provide an inherent order over all subgraphs in terms of the maximal number of copies of such subgraphs in any hidden-variable model, and provide a method to compare the denseness of subgraphs created by any degree distribution to its absolute maximum provided by our bounds.
    We also show that the Chung-Lu model achieves the maximal number of all types of subgraphs among all hidden-variable models.
    \item[(iii)] We compare the extremal graph model that provides the highest subgraph counts to existing results for power-law random graphs. 
We show that for all subgraphs, a power-law degree distribution with $\tau\leq 2$ achieves the maximal subgraph scaling predicted by our bounds. This shows that power-laws with $\tau\leq 2$ are densest networks in terms of subgraph counts among all networks with the same average degree and MAD. Power laws with larger exponents ($\tau>2$) do not achieve the maximal scaling, and the subgraph counts scale at a slower rate. This shows that in the sparse setting ($\tau>2$), where the average degree does not grow, power-law networks are not `optimal' in terms of subgraph counts, whereas in the dense setting ($\tau\leq 2$) they are.
\item[(iv)] We demonstrate that the robust bounds indeed bound the number of subgraphs in nine real-world data sets. This analysis does not require any assumption on the degree distribution of these data sets; only knowledge of the average degree, the MAD, and the maximal degree is required. In particular, the robust bounds work for both power-law and non power-law distributed data. 
\end{itemize}



We introduce the hidden-variable model and assumptions on the degree distribution in Section~\ref{sec:model}. We then solve the maximization problem that finds the extremal random graph that generates the maximal subgraph counts in Section~\ref{sec:extreme}. The scaling laws for subgraph counts as function of the network size are presented in Section~\ref{sec:scale}. In Section~\ref{sec:var} we obtain some results for the setting when the variance instead of the MAD is known. In Section~\ref{sec:powerlaw}, we compare this extremal random graph and tight subgraph bounds with existing results for scale-free networks with power-law degrees.
Section~\ref{sec:data} shows subgraph counts and bounds for nine real-world networks. The paper is concluded in Section~\ref{sec:out} with a discussion and outlook.

\section{hidden-variable model}\label{sec:model}
As a random graph model, we employ the hidden-variable model, in which every vertex $i\in[n]$ has a weight $h_i$. Traditionally, one then assumes that the weights $h_1,\ldots,h_n$ are independent and follow some given distribution. In this paper, however, we only specify partial information about the weight (i.e.~degree) distribution. We will assume that for the weights we know the minimal and maximal value, the mean and the mean absolute deviation (MAD). Let $h$ denote a generic weight. 
Then we assume that the weights are 
sampled independently from a degree distribution such that (i) $h=h_i$  has support $\supp(h)=[a,h_c]$ with $-\infty<a\leq h_c<\infty$,  (ii) $\Exp{h}=\mu$ and (iii) $\Exp{|h-\mu |}=d$. This defines the ambiguity set 
\begin{align}\label{eq:psupp}
	&\mathcal{P}(\mu,d)=\nonumber\\
	&\{\pprob:  \supp(h)\subseteq [a,h_c], \Exp{h}=\mu ,\Exp{|h-\mu|}=d\}.
\end{align}
Hence, when we now analyze the hidden-variable model under the assumption that the weight distribution belongs to $\mathcal{P}(\mu,d)$ we perform a distribution-free analysis of the random graph model. 

 We further assume that every pair of vertices is connected independently with probability
\begin{equation}\label{eq:pconpf}
	p(h_i,h_j)=f(h_i h_j/h_s^2).
\end{equation} 
 where $h_s$ is the structural cutoff. This shape of the connection probability ensures that the weight of a vertex is close to its degree~\cite{stegehuis2017b}. The structural cutoff describes the maximal degree of vertices that are not prone to degree-degree correlations~\cite{boguna2004}. As soon as the degree of a vertex becomes larger than the structural cutoff, it is forced to connect to lower degree vertices, as only few high degree vertices can be present while keeping the average degree fixed. The structural cutoff has mainly been investigated for power-law networks, where $h_s\sim\sqrt{\mu n}$~\cite{boguna2004,colomer2012,catanzaro2005}. The natural cutoff describes the constraint on the largest possible network degree. If the objective is to generate uncorrelated networks, this natural cutoff should be smaller than the structural cutoff, as larger vertices experience degree-correlations. Therefore, the natural cutoff is often assumed to be equal to the structural cutoff of $\sqrt{\mu n}$~\cite{catanzaro2005,bianconi2006,chung2002}. In this setting, many network properties can be related to moments of the hidden variable distribution, which makes it possible to investigate subgraph counts~\cite{bianconi2006}, distances~\cite{chung2002} and clustering~\cite{colomer2012}. We henceforth assume the setting where $h_s=h_c$, equal structural and natural cutoff, so that the generated networks are uncorrelated~\cite{boguna2004}.

 
 For the connection function $f$ we make the following assumption: 
 \begin{assumption}\label{ass:functionf1}\leavevmode
 		$f(x)\geq 0$ is non-decreasing and convex for $x\in[0,1]$. \label{ass:f1}
 \end{assumption}

The class of hidden-variable models satisfying Assumption~\ref{ass:functionf1} is very rich. In particular, it contains the following three 
frequently used connection probabilities: The Chung-Lu model~\cite{chung2002,bollobas2007,hofstad2017b} 
\begin{equation}\label{eq:ex1}
f(u)=\min\{u,1\},
\end{equation}
the Poisson random graph~\cite{bollobas2007,norros2006,bhamidi2010}
\begin{equation}\label{eq:ex2}
f(u)=1-\textup{e}^{-u}
\end{equation}
and
the generalized random graph 
\cite{britton2006,park2004,colomer2012,squartini2011}
\begin{equation}\label{eq:ex3}
f(u)=\frac{u}{1+u}.
\end{equation}
 
 
\section{Extremal random graph}\label{sec:extreme}
We now investigate the maximal subgraph counts in hidden-variable models over all probability distributions of the weights that satisfy~\eqref{eq:psupp}. 
We thus investigate the number of copies $N_H$ of a given subgraph $H=({V}_H,{E}_H)$ on $k$ vertices. Let the degrees of the vertices in $H$ be denoted by $d_1,d_2,\dots,d_k$. Then, the expected number of copies of $H$ becomes
\begin{align}
&	\mathbb{E}[N_H\mid (h_i)_{i\in[n]}]=\sum_{i_1<i_2<\dots<i_k}\prod_{\{u,v\}\in {E}_H}p(h_{i_u},h_{i_v})\nonumber\\
	& =\frac{1}{\Aut(H)}\sum_{i_1=1}^n\sum_{i_2=1}^n\dots \sum_{i_k=1}^{n}\prod_{\{u,v\}\in {E}_H}f\Big(\frac{h_{i_u}h_{i_v}}{h_s^2}\Big)
\end{align}
where $\Aut(H)$ denotes the number of automorphisms of $H$.
Therefore,
\begin{align}\label{eq:ENH}
	\Expp{N_H}&=\frac{n^k}{\Aut(H)}\Expp{\prod_{\{u,v\}\in {E}_H}f\Big(\frac{h_{u}h_{v}}{h_s^2}\Big)},
\end{align}
where $h_i$ denote independent copies of the random variable $h$.

\begin{lemma} Under {\rm Assumption~\ref{ass:functionf1}},
the function $\prod_{\{u,v\}\in \mathcal{E}_H}f(\frac{h_{u}h_{v}}{h_s^2})$ is convex in all $h_i$.
\end{lemma}
\begin{proof}
The function $f$ is convex, non-decreasing and positive on $[0,1]$, so a product of these $f$-functions is also convex. 
\end{proof}

Hence, under the assumptions made for the hidden-variable model, the subgraph count $\Expp{N_H}$ viewed as function of the hidden variables is convex in $(h_1,\ldots,h_n)$. This convexity can be leveraged to employ a method from distributionally robust optimization for maximizing a convex function under certain mean-MAD-range constraints:
\begin{theorem}[Extremal graph model]\label{thm1h}
Under {\rm Assumption~\ref{ass:functionf1}},
the extremal distribution that solves 
$
\max_{\mathbb{P} \in \mathcal{P}_{(\mu,d)}} \Expp{N_H}$
consists for each $h_i$ of a three-point distribution with values ${a,\mu,h_c}$ 
and probabilities
\begin{align}\label{eq:3pointp}
	&p_a=\frac{d}{2(\mu-a)},\quad p_\mu=1-\frac{d}{2(\mu-a)}-\frac{d}{2(h_c-\mu)},\nonumber\\
	& p_{h_c}=\frac{d}{2(h_c-\mu)}.
\end{align}
\end{theorem}
Theorem~\ref{thm1h} follows from the the general upper bound in \cite{BenTal1972} on the expectation of a convex function of independent random variables with mean-MAD-range information. The proof of Theorem~\ref{thm1h}  crucially relies on the fact that the solution of the univariate case can be straightforwardly extended to the multivariate case. Consider 
 a univariate convex function $x\mapsto g(x)$ and consider 
 the maximum $\Expp{g(h)}$ with $h$ a random variable on $[a,h_c]$ with mean $\mu$ and MAD $d$. In this univariate setting we can formulate this maximization problem as 
\begin{equation}\label{test3}
\begin{aligned}
&\max_{\pprob(x)\geq0} &  &\int_x g(x){\rm d} \pprob(x)\\
&\text{s.t.} &      & \int_x |x-\mu|{\rm d}\pprob(x)=d, \int_x x{\rm d}\pprob(x)=\mu,\\
&  & & \int_x {\rm d}\pprob(x)=1,   
\end{aligned}
\end{equation}
a semi-infinite linear program with three equality constraints, and as solution for $\pprob=\pprob(x)$ the three-point distribution on $\{a,\mu, h_c\}$. Notice that this solution does not depends on the specific shape of the convex function $g(x)$, which makes the multivariate counterpart of \eqref{test3} equally tractable. 
Take as an example the function $g(h_1,\ldots, h_n):= \Expp{N_H}$, and formulate the maximal subgraph count as 
\begin{equation}\label{test4}
\begin{aligned}
&\max_{\pprob(x)\geq0} &  &\int_x g(x_1,\ldots, x_n){\rm d} \pprob(x_1)\cdots{\rm d}\pprob(x_n)\\
&\text{s.t.} &      & \int_x |x-\mu|{\rm d}\pprob(x)=d, \int_x x{\rm d}\pprob(x)=\mu,\\
& & & \int_x {\rm d}\pprob(x)=1.   
\end{aligned}
\end{equation}
Indeed, to deal with this multivariate case,
we can recursively apply the univariate result. Suppose we first apply this result to $h_1$. Then the worst-case distribution is as in Theorem~\ref{thm1h}, independent of the values for $h_2,\ldots,h_n$. Moreover, the maximal expectation  becomes a convex function in $h_2,\ldots,h_n$, since the extremal probabilities for $h_1$ are nonnegative. Hence, we can apply the result above for the univariate case to $h_2$, and so on, which then establishes Theorem~\ref{thm1h}. 
The theorem thus shows that the extremal random graph for all possible subgraphs and all possible hidden-variable models satisfying Assumption~\ref{ass:functionf1} consists of vertices with only three degrees: $a$, $\mu$ and $h_c=h_s$. Under the canonical choice $h_s=\sqrt{\mu n}$, 
\begin{equation*}
	p_a=\frac{d}{2(\mu-a)},\quad p_\mu\approx1-\frac{d}{2(\mu-a)},\quad p_{\sqrt{\mu n}}\approx\frac{d}{2\sqrt{\mu n}}.
\end{equation*}
As $p_a$ and $p_\mu$ are constant in $n$, the extremal random graph contains $O(n)$ vertices of low degrees $a$ and $\mu$. Furthermore, $p_{\sqrt{\mu n}}$ scales as $1/\sqrt{n}$, creating on average $d\sqrt{n}/(2\mu)$ vertices with degrees as large as $\sqrt{\mu n}$. The connection probability~\eqref{eq:pconpf} shows that $p(h_s,h_s)=f(1)$, which does not depend on $n$. This means that the $\sqrt{\mu n}$ weight vertices form a dense subgraph. The denseness is controlled by the parameter $f(1)$. When $f(1)=1$, these vertices form a clique, and for $f(1)=p<1$, these high-degree vertices form a dense Erd\H{o}s-R\'enyi random graph with probability $p$. On the other hand, vertices with weight $a$ or $\mu$ in the extremal random graph have finite average degree. This shape of extremal random graph is illustrated in Figure~\ref{fig:extremal} for the case $f(1)=1$.

\begin{figure}[tb]
    \centering
    \centering
    \includegraphics[width=0.6\linewidth]{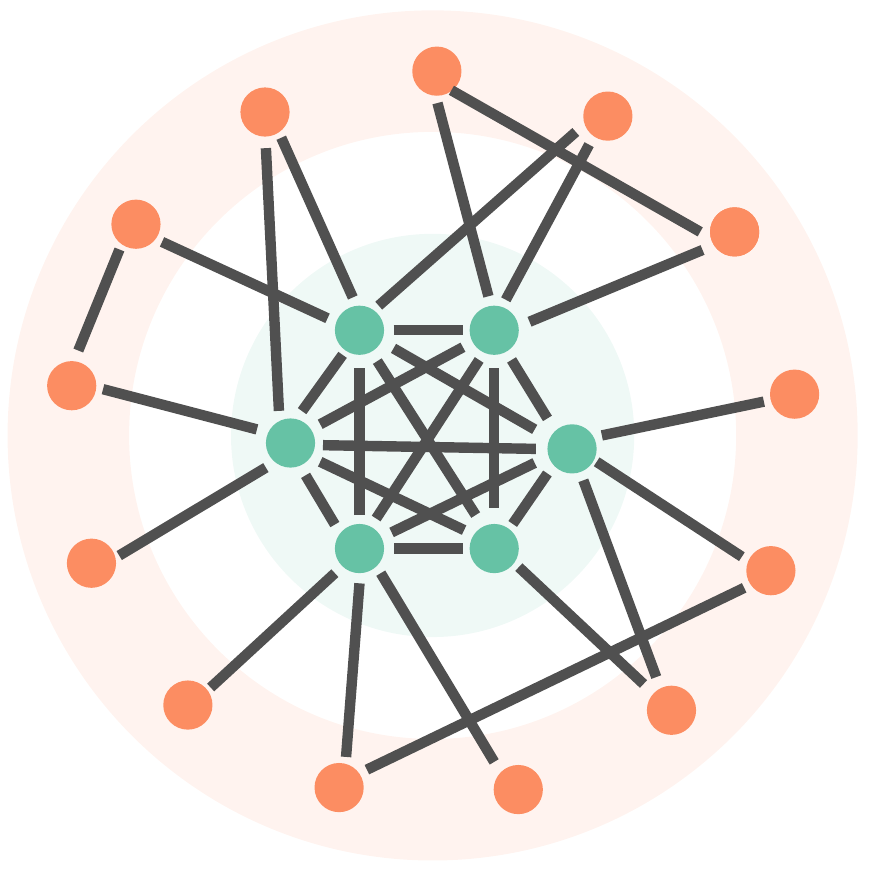}
    \caption{Illustration of the extremal Chung-Lu random graph: a clique of size $O(\sqrt{n})$ (green), and $O(n)$ vertices of weights $\mu$ or $a$ with small expected degree (orange).}
    \label{fig:extremal}
    \end{figure}

\section{Scaling laws for large networks}\label{sec:scale}
 A direct consequence of Theorem \ref{thm1h} is that the tight bound on the subgraph count is obtained by enumerating over all $3^{n}$ permutations of the outcomes $\{a,\mu, h_c\}$ for all weight $h_1,\ldots,h_n$. This gives the following result:
\begin{corollary}\label{corollary1}
Under {\rm Assumption~\ref{ass:functionf1}}, the tight bound on the subgraph count can be expressed as 
\begin{align}\label{eq:NHgen}
	&\max_{\mathbb{P} \in \mathcal{P}_{(\mu,d)}} \Expp{N_H}=\frac{n^k}{\Aut(H)}\nonumber\\
	& \times \sum_{i_1\in\{a,\mu,h_c\}}\!\!\dots \!\!\sum_{i_k\in\{a,\mu,h_c\}}\prod_{j=1}^kp_{i_j}\prod_{\{s,t\}\in E_H}f\Big(\frac{i_si_t}{h_c^2}\Big).
\end{align}
\end{corollary}
We next show how the maximal subgraph counts scale as function of the network size. To obtain these scaling relations, we employ \eqref{eq:NHgen}
and make an additional assumption next to Assumption~\ref{ass:functionf1} on the connection probabilities:  
 \begin{assumption}\label{ass:functionf2}\leavevmode
 	$f(x)=xr(x)$ where $r(0)=1$ and $r(x)$ decreases in $x$. \label{ass:f2}
 \end{assumption}
 While Assumption~\ref{ass:functionf2} is a more detailed assumption on the connection probabilities than Assumption~\ref{ass:functionf1}, it still contains the three classical examples of hidden-variable models in  \eqref{eq:ex1}-\eqref{eq:ex3}.


We now investigate the behavior of~\eqref{eq:NHgen} when $f$ satisfies Assumption~\ref{ass:functionf2}. In that case, when $h_s=h_c$,
\begin{align}\label{eq:ENHsplit}
	& \Expp{N_H}=\frac{n^k}{\Aut(H)}\nonumber\\
	& \times \sum_{i_1\in\{a,\mu,h_c\}}\!\!\dots \!\! \sum_{i_k\in\{a,\mu,h_c\}}\prod_{j=1}^kp_{i_j}\prod_{\{s,t\}\in E_H}\frac{i_si_t}{h_c^2}r\Big(\frac{i_si_t}{h_c^2}\Big)\nonumber\\
	& =\frac{n^k}{\Aut(H)h_s^{2E_H}}\nonumber\\
	& \times \!\!\! \sum_{i_1\in\{a,\mu,h_c\}}\!\!\dots\!\! \sum_{i_k\in\{a,\mu,h_c\}}\prod_{j=1}^kp_{i_j}i_j^{d_j}\!\! \prod_{\{s,t\}\in E_H}\!\! r\Big(\frac{i_si_t}{h_c^2}\Big).
\end{align}
We show in Appendix~\ref{app:scalingMAD} that for every vertex $j$ with $d_j\geq 2$, the summation over $i_j\in\{a,\mu,h_c\}$ is dominated by the term containing $h_c$, so that the other terms may be ignored. 
Therefore, for subgraphs with minimal degree at least 2, we can ignore the terms in~\eqref{eq:ENHsplit} with $i_j=a$ or $i_j=\mu$, yielding 
\begin{align}
	\Expp{N_H}
	& \sim \frac{n^k}{\Aut(H)h_s^{2E_H}}\prod_{j=1}^k\frac{d}{2(h_c-\mu)}h_c^{d_j}\prod_{\{s,t\}\in E_H}r(1)\nonumber\\
	& =\frac{n^kd^kh_c^{2E_H-k}}{h_c^{2E_H}2^k\Aut(H)}r(1)^{E_H}.
\end{align}

When $d_j=1$, the contributions of $i_j=a,\mu,h_c$ in~\eqref{eq:ENHsplit} are of similar order of magnitude. In Appendix~\ref{app:scalingMAD} we show that this yields the following result:

\begin{theorem}[Core structure]\label{thm:subgraphsmad}
	Let $H=(V_H,E_H)$ be a connected subgraph on $k$ vertices and make {\rm Assumption~\ref{ass:functionf1}}~and~{\rm Assumption \ref{ass:functionf2}}. 
	\begin{itemize}[label=(\roman*)]
	\item[{\rm (i)}]
	When $d_H>1$ for all $v\in V_H$, $h_s=h_c$ and $h_c\to\infty$ as $n\to\infty$, 
	\begin{equation}\label{eq:ENHasymptotic}
		\frac{\max_{\pprob\in\mathcal{P}(\mu,d)}\Expp{N_H}}{n^kh_c^{-k}}\to\frac{d^k}{2^k\textup{Aut}(H)}r(1)^{|E_H|}.
	\end{equation}
\item[{\rm (ii)}] 
 When $k\geq 3$ and $h_c\to\infty$ and $h_s=h_c$ as $n\to\infty$, 
\begin{align}\label{eq:ENHasymptoticdeg1}
	& \frac{\max_{\pprob\in\mathcal{P}(\mu,d)}\Expp{N_H}}{n^kh_c^{-k}}\nonumber\\
	& \to\frac{d^{k-n_1}}{\Aut(H)2^{k-n_1}}\left(\frac{d}{2}(r(1)-1)+\mu\right)^{n_1}r(1)^{E_H-n_1},
	\end{align}
\end{itemize}
where $n_1$ denotes the number of degree-1 vertices in $H$.
\end{theorem}

This theorem shows subgraphs with minimal degree at least 2 predominantly appear in the core of the extremal random graph containing all vertices of weight $h_c$. Other subgraphs that contain degree-one vertices, asymptotically have their degree-one vertices everywhere in the extremal random graph, while the other subgraph vertices typically still only appear in the core of the extremal random graph, see Figure~\ref{fig:extremalsubgraph}. 

Furthermore,~\eqref{eq:ENHasymptoticdeg1} reveals the interesting property that the scaling in $n$ and $h_c$ is the same for all possible subgraphs of size $k$. The effect of the precise subgraph structure is only visible in the leading order constant. Also, if we compare all subgraphs of a given size $k$ with minimal degree at least 2,~\eqref{eq:ENHasymptotic} shows that this leading order constant only depends on the subgraph through its number of automorphisms. Therefore, we can easily order all such subgraphs based on the maximal number of times they appear in any hidden-variable model. For example, among all subgraphs of a given size, cliques have the largest number of automorphisms. Thus, $k$-cliques appear the least often among all subgraphs of size $k$ in the extremal random graph, and therefore have the lowest upper bound on their count among all size-$k$ subgraphs. The rest of the ordering in terms of maximal number of subgraphs, is only determined by the number of subgraph automorphisms in decreasing order.  

\begin{figure}[tb]
    \centering
    \includegraphics[width=0.6\linewidth]{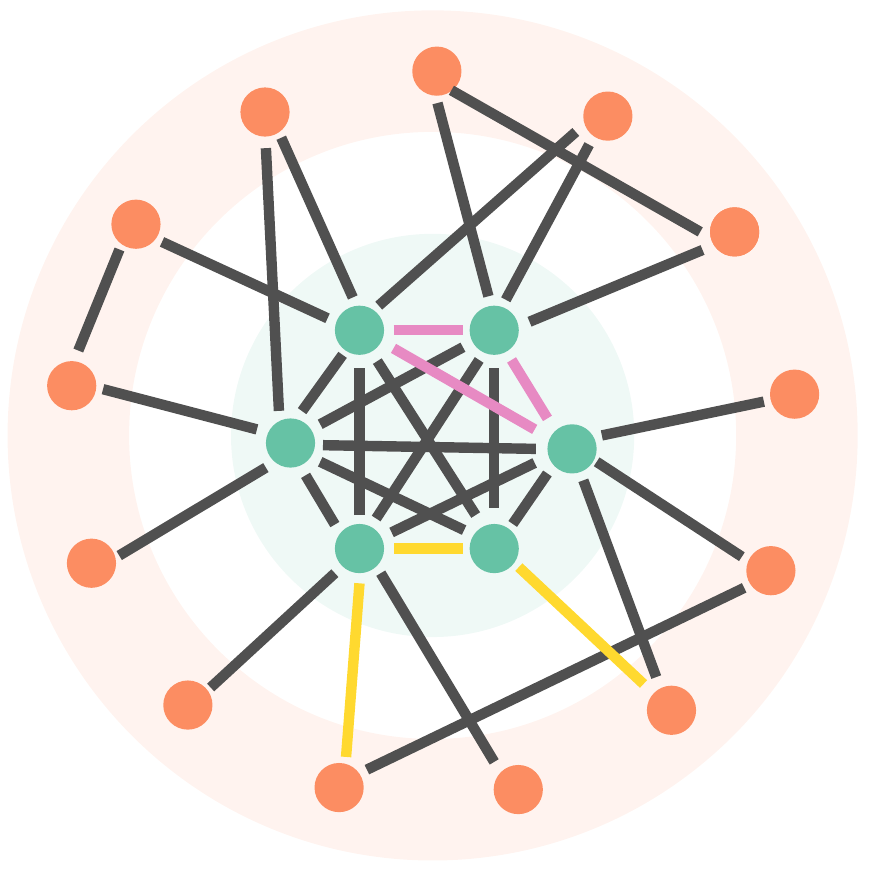}
    \caption{Subgraphs with minimal degree at least 2 asymptotically appear on the $\sqrt{n\mu }$ degree vertices (pink triangle), while degree-1 vertices inside a subgraph typically appear at the degree-$\mu$, $a$ or $h_c$ vertices (yellow 4-path).}
    \label{fig:extremalsubgraph}
\end{figure}

When degree-one vertices are contained in the subgraph, \eqref{eq:ENHasymptoticdeg1} shows that a combination of the number of subgraph automorphisms, a term containing the fixed parameters $d$ and $\mu$, the model-specific term $r(1)$ and the number of degree-one vertices in the subgraph. When we take the specific case $r(1)=1$ and $h_s=h_c=\sqrt{\mu n}$,~\eqref{eq:ENHasymptoticdeg1} simplifies to
\begin{align}\label{eq:ENHcl}
	\Expp{N_H}
	& \sim \frac{n^{k/2}d^{k-n_1}}{\Aut(H)2^{k-n_1}\mu^{k/2-n_1}}.
\end{align}

Equation~\eqref{eq:ENHcl} shows that for all subgraphs of size $k$ the only effect of the graph structure of $H$ is determined by the number of automorphisms of $H$ and the constant $(2\mu/d)^{n_1}$. Figure~\ref{fig:motif4} illustrates the ordering in terms of maximal frequency of all subgraphs of size 4 by showing the leading order constant of~\eqref{eq:ENHcl}. A lower constant in Figure~\ref{fig:motif4} therefore indicates that this particular subgraph appears less frequently than a subgraph with a higher constant. Thus, the clique is the least frequently occurring subgraph. The most frequently appearing subgraph depends on the constant $c=2\mu /d>1$, and is either the path or the claw. A similar ordering for size 5 subgraphs is provided in Appendix~\ref{sec:motif5order}.

\begin{figure}[tb]
    \centering
    \includegraphics[width=0.48\textwidth]{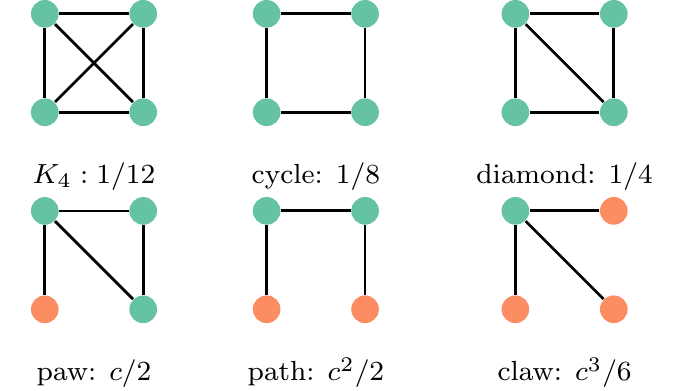}
    \caption{The leading constant $c^{n_1}/\Aut(H)$ for the maximal scaling of the number of subgraphs in $n$ for all subgraphs on 4 vertices, where $c=2\mu/d$. }
    \label{fig:motif4}
\end{figure}

We now apply Theorem~\ref{thm:subgraphsmad} to the three classical hidden-variable models mentioned in Section~\ref{sec:model}. In particular, we show that  the Chung-Lu model can generate the largest amount of cliques. 
\begin{proposition}
    Suppose $h_s=h_c=\sqrt{\mu n}$. The Chung-Lu model can generate the maximal number of cliques among all connection probabilities that fall under {\rm Assumption~\ref{ass:functionf1}}, scaling as
    \begin{align}\label{eq:EKkcl}
	\frac{\max_{\pprob\in\mathcal{P}(\mu,d)}\Expp{N_{K_k}}}{n^{k/2}}& \to\frac{d^k}{k!2^k\mu^{k/2}},
    \end{align}
        Furthermore, all three classical hidden-variable models \eqref{eq:ex1}-\eqref{eq:ex3} satisfy
    \begin{equation}
        \frac{\max_{\pprob\in\mathcal{P}(\mu,d)}\Expp{N_{K_k}}}{n^{k/2}}\to \frac{d^kr(1)^{k(k-1)/2}}{k!2^k\mu^{k/2}}.
    \end{equation}
\end{proposition}

\begin{proof}
When taking $h_s=h_c=\sqrt{\mu n}$, for all connection probabilities satisfying Assumption~\ref{ass:functionf1},~\eqref{eq:ENHasymptotic} gives for the expected number of $k$-cliques 
\begin{equation}\label{eq:Kkexp}
	\max_{\pprob\in\mathcal{P}(\mu,d)}\Expp{N_{K_k}}\approx \frac{d^kn^{k/2}}{k!2^k\mu^{k/2}}r(1)^{k(k-1)/2},
\end{equation}
proving the second part of the statement. Furthermore, from~\eqref{eq:Kkexp} it is not difficult to see that the Chung-Lu model can generate the maximal number of cliques among all connection probabilities that fall under Assumption~\ref{ass:functionf1}. Indeed, the function $r(x)=1$ is the maximal possible function under Assumption~\ref{ass:functionf1}, and~\eqref{eq:Kkexp} is increasing in $r(1)$, proving the first part of the statement.
\end{proof}

\section{Fixing variance instead of MAD}\label{sec:var}
We have shown that under mean-MAD-range information, the search for the extremal random graph that maximizes subgraph counts was tantamount to solving the optimization problem in \eqref{test4}. The three-point solution of the $1$-dimensional problem in \eqref{test3} carried over directly to the $n$-dimensional problem in \eqref{test4}. We now ask what happens when we replace the MAD information with variance information. 

We then first solve the 1-dimensional counterpart of \eqref{test3}, but now optimizing over all distributions with known range, mean and variance (instead of MAD). As it turns out, the solution (the extremal distribution) will depend on the convex function $g(x)$, which has severe consequences for the multivariate case, i.e., when we consider $g(h_1,\ldots,h_n)$. In that case, 
the extremal distribution depends on the values of $h_2,\ldots,h_n$, and calculating (in closed form) the extremal distribution as a function of $h_2,\dots,h_n$ seems to be impossible. This means that for each $n$, a new $n$-dimensional problem needs to be solved, implying that the extremal random graph is different for every value of $n$. 

We now show how the key result for mean-MAD ambiguity, Theorem~\ref{thm1h}, can be used to obtain results for mean-variance ambiguity. In general, MAD and variance are related as \cite{BenTal1985,eekelen2019}
$$
d_{\min}:=\frac{2 \sigma^2}{h_c-a}\leq d\leq \sigma=:d_{\max}. 
$$
Let $\mathcal{P}^*_{(\mu,\sigma)}$ denote the ambiguity set that contains all distributions with known range, mean and variance, i.e.,
\begin{align*}
& \mathcal{P}^*_{(\mu,\sigma)} = \nonumber\\
& \left\{ \mathbb{P}: \ \text{supp}(X) \subseteq [a,h_c], \ \mathbb{E}_{\mathbb{P}}(X) = \mu, \ \mathbb{E}_{\mathbb{P}}( X - \mu )^2 = \sigma^2\right\}. 
\end{align*}
Since $\max_{\mathbb{P} \in \mathcal{P}_{(\mu,d)}} \Expp{g(h_1,\ldots,h_n)}$ is non-decreasing in $d$, see \cite{eekelen2019}, we obtain for fixed $\sigma$ the bounds
\begin{align}\label{cvb}
&\max\limits_{\mathbb{P} \in \mathcal{P}_{(\mu,d_{\rm min})}} \Expp{g(h_1,\ldots,h_n)}\leq \max\limits_{\mathbb{P} \in \mathcal{P}^*_{(\mu,\sigma)}} \Expp{g(h_1,\ldots,h_n)}\nonumber\\
& \leq \max\limits_{\mathbb{P} \in \mathcal{P}_{(\mu,d_{\rm max})}} \Expp{g(h_1,\ldots,h_n)}.
\end{align}
Thus, when fixing the variance at $\sigma^2$, we need to consider the range of $d\in[2\sigma^2/(h_c-a),\sigma]$. 
For the extremal random graph of~\eqref{eq:3pointp}, 
\begin{align}
    \sigma^2&=\frac{da^2}{2(\mu-a)}-\frac{d\mu^2}{2(\mu-a)}-\frac{d\mu^2}{2(h_c-\mu)}+\frac{dh_c^2}{2(h_c-\mu)}\nonumber\\
    & = d(h_c-a)/2.
\end{align}
Thus, the lower bound, $d_{\min}=2\sigma^2/(h_c-a)$ ensures that the extremal graph has the desired variance $\sigma^2$.
In this section we therefore set $d=2\sigma^2/(h_c-a)$. Then, the MAD extremal random graph with this value of $d$ gives the desired variance, and provides a lower bound on the maximal number of subgraphs with given variance. In particular, we are interested in the setting where $\sigma^2\to \infty$ as $n\to\infty$ but $\sigma^2/h_c\to 0$, similar to observations in many real-world networks. 

In Appendix~\ref{app:MADvar}, we again consider~\eqref{eq:ENHsplit}, and investigate which of the three terms in the summation ($a$, $\mu$ or $h_c$) has the dominant contribution for large $n$. The dominating terms are slightly different compared to the fixed $d$ setting: when $d_j\geq 3$, the contribution from $i_j=h_c$ dominates, whereas for $d_j=2$, the contribution from $i_j=h_c$ and from $i_j=\mu$ have the same order of magnitude. Finally, when $d_j=1$,  the contribution to~\eqref{eq:ENHsplit} from $i_j=\mu$ dominates the other contributions. Appendix~\ref{app:MADvar} shows that this gives the following result:

\begin{theorem}[Diminishing $d$]\label{thm:maxNhsigma}
	
	Let $H=(V_H,E_H)$ be a connected subgraph on $k$ vertices. Let $d=2\sigma^2/(h_c-a)$ where $h_s=h_c\to\infty$ and $\sigma^2/h_c\to 0$. Then, under {\rm Assumption~\ref{ass:functionf1}}~and~{\rm \ref{ass:functionf2}},  
		\begin{align}\label{eq:NHdsmall}
			& \frac{\max_{\pprob\in\mathcal{P}(\mu,d)}\Expp{N_H}}{n^kh_c^{-2k+n_1}} \to  \frac{r(1)^{E_{\geq 3,\geq 3}}}{\Aut(H)}\left(\sigma^2r(1)+\mu^2\right)^{n_{2,1}}\nonumber\\
			& \times \left(\sigma^2r(1)^2+\mu^2\right)^{n_2-n_{2,1}}\mu^{n_1}\sigma^{2n_{\geq 3}} ,
		\end{align}
	where $n_i$ and $n_{\geq i}$ denote the number of vertices of degree $i$ or degree at least $i$ in $H$, and $E_{\geq 3,\geq 3}$ denotes the number of edges between vertices of degree at least 3 in $H$. Furthermore, $n_{2,1}$ denotes the number of degree-2 vertices in $H$ that are connected to a degree-1 vertex.
\end{theorem}

This theorem again shows that in the extremal random graph with $d=2\sigma^2/(h_c-a)$ subgraphs counts are dominated by subgraph counts in specific formations of the extremal random graph. For example, subgraphs with minimal degree at least three almost exclusively appear in the core of $\sqrt{\mu n}$ vertices. Degree-one vertices in a subgraph on the other hand, typically appear at the vertices of degree $\mu$.

As an example, take the Chung-Lu model with $h_s=\sqrt{\mu n}$, $a=1$ and $p(h,h')=\min(hh'/(\mu n),1)$. Then \eqref{eq:NHdsmall} gives for $k>3$,
\begin{align}\label{eq:EKkvar}
	\Expp{N_{K_3}}
	& \approx \frac{1}{6\mu^{3}}\left(\mu^2+\sigma^2\right)^3,\nonumber\\
		\Expp{N_{K_k}}&\approx  \frac{\sigma^{2k}}{k!\mu^k}, \quad k>3.
\end{align}
This is an intuitive result, because when maximizing $\mathbb{E}[X^2]$ while keeping the variance and the mean degree fixed, we expect to end up with some function of $\mathbb{E}[X^2]$.

For other subgraphs $H$ we can obtain similar results. 
When $H$ has $s_1$ vertices of degree 1, $s_2$ of degree 2 and denote $s_{\geq 3}$ vertices with degree at least 3,~\eqref{eq:NHdsmall} becomes
\begin{align}\label{eq:ENhvar}
	\Expp{N_{H}}
	& \approx  \frac{n^{s_1/2}}{\Aut(H)\mu^{k-3/2s_1}}\left(\mu^2+\sigma^2\right)^{s_2}\sigma^{2s_{\geq 3}}.
\end{align}
Thus, the more degree-1 vertices a subgraph has, the more often it appears, as a scaling in $n$. When the number of degree-1 vertices remains unchanged, the scaling in $n$ of the subgraph count remains the same. However, having more degree-2 vertices increases the leading order term of the subgraph count. And as before, subgraphs with more automorphisms appear less often.

\section{Power-law random graphs}\label{sec:powerlaw}


We now compare the results on the maximal clique counts among all weight distributions with the subgraph counts in the frequently-used power-law weights, to answer the question of how close power-law degrees are to the extremal random graph. Thus, we assume a power-law distribution with cutoff for the weights 
\begin{equation}
	\pprob(h)=Ch^{-\tau},
\end{equation}
for $h\in[1,h_c]$ and some $\tau$ and $C$. Then,
\begin{equation}\label{eq:pld}
	d=\frac{C(2\mu^{2-\tau}-1-h_c^{2-\tau}))}{\tau-2}+\frac{C\mu(-2\mu^{1-\tau}+1+h_c^{1-\tau}))}{\tau-1}.
\end{equation}
Observe that for $\tau>2$, $d$ is approximately constant, whereas for $\tau<2$ it grows as $h_c^{2-\tau}$ (as there $\mu$ also grows as $h_c^{2-\tau}$). We split the comparison into two classes for $\tau$: so called sparse ($2<\tau<3$) and dense ($1<\tau<2$) scale-free networks. 

{\it Sparse scale-free networks}. 
We now compare the number of subgraphs in power-law random graphs in the regime $2<\tau<3$ with the extremal number of subgraphs among all random graphs with equal mean and MAD. When $2<\tau<3$, under a cutoff at $h_s=h_c=\sqrt{\mu n}$, the expected number of cliques in a power-law random graph with degree-exponent $\tau$ equals~\cite[Eq.~(1.7)]{janssen2019}
\begin{equation}\label{eq:Ekkcutoffpl}
	\expec_{pl}[N_{K_k}]\approx \frac{ n^{k/2(3-\tau)}\mu^{k/2(1-\tau)}}{k!}\left(\frac{C}{k-\tau}\right)^k.
\end{equation}
The MAD-maximizer with the same $\mu$ and $d$ as this power-law distribution for $2<\tau<3$ becomes according to~\eqref{eq:EKkcl}
\begin{align}
    &\max_{\mathbb{P} \in \mathcal{P}_{(\mu,d)}} \Expp{N_{K_k}}\nonumber\\
    & \approx \frac{n^{k/2}}{k!2^k\mu^{k/2}}\Big(\frac{C(2\mu^{2-\tau}-1)}{\tau-2}-\frac{C(2\mu^{2-\tau}-\mu)}{\tau-1}\Big)^k.
\end{align}
This scaling in $n$ of $n^{k/2}$ is much larger than the scaling of $n^{k/2(3-\tau)}$ in power-law random graphs. 
Thus, the extremal random graph is asymptotically much more dense than power-law random graphs in terms of cliques. 

\begin{figure*}[tb]
    \centering
    \subfloat[$\tau=2.5$]{
    \includegraphics[width=0.48\textwidth]{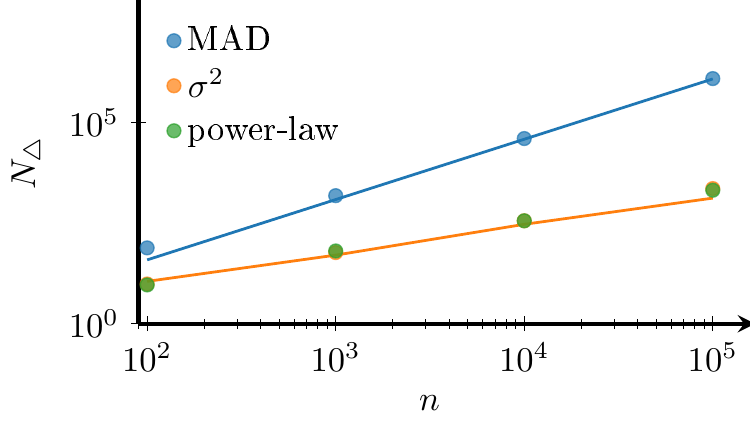}
    \label{fig:triang25}
    }
    \hfill
    \subfloat[$\tau=1.5$]{
    \includegraphics[width=0.48\textwidth]{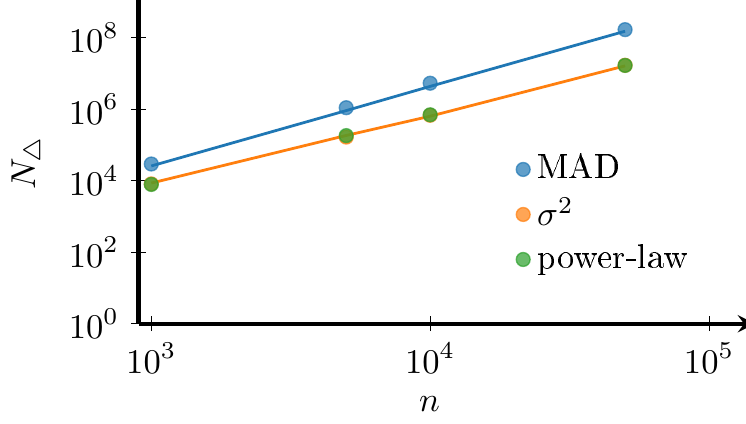}
    \label{fig:triang15}
    }
    \caption{The number of triangles against the network size $n$ in: a power-law random graph, an MAD-optimal model with same $\mu$ and $d$ as the power-law model, an MAD-optimal model with equal $\mu $ and $\sigma^2$ as the power-law model. Dots are simulated values, whereas the blue and the orange curves are the MAD-optimal predictions from~\eqref{eq:EKkcl} and~\eqref{eq:EKkvar}, respectively.}
    \label{fig:simpl}
\end{figure*}

Let us then compare the expected number of cliques in a power-law random graph~\eqref{eq:Ekkcutoffpl} with the extremal number of subgraphs with $d=2\sigma^2/(\sqrt{\mu n}-1)$, where $\mu$ and $\sigma^2$ denote the mean and variance of the power-law distribution. This ensures that the variance of the MAD extremal random graph is equal to the variance of the power-law distribution. 
The variance of power-laws with cutoff at $\sqrt{\mu n}$ equals $\sigma^2=\frac{C}{3-\tau}\sqrt{\mu n}^{3-\tau}$. Then,~\eqref{eq:EKkvar} yields for the variance-maximal number of cliques that
\begin{equation}\label{eq:Ekkcutoff}
	\max_{\mathbb{P} \in \mathcal{P}_{(\mu,\frac{\sigma^2}{\sqrt{\mu n}-1})}}\Expp{N_{K_k}}= \frac{ n^{k/2(3-\tau)}\mu^{k/2(1-\tau)}}{k!}\left(\frac{C}{3-\tau}\right)^k.
\end{equation}
When comparing the leading order constant to the one computed for power-law random graphs with cutoff at $\sqrt{\mu n}$, it exactly agrees for $k=3$, so for triangles. Therefore,~\eqref{eq:Ekkcutoff} suggests that power-law random graphs contain the maximal amount of triangles among all Chung-Lu models with the same variance. In other words: power-law random graphs are the most clustered random graphs among all Chung-Lu models with given variance. For larger cliques, the constant in~\eqref{eq:Ekkcutoff} is higher than the one for the power-law random graph with cutoff. Thus, power-law random graphs do not contain the maximal amount of larger cliques among all Chung-Lu random graphs with the same variance in terms of leading order constant. The scaling in $n$, $n^{k/2(3-\tau)}$ however, still agrees between the power-law number of cliques and the variance-extremal random graph, so that in order of magnitude, power-law random graphs achieve the largest possible number of cliques among all graphs with the same variance. Proving this however, needs the proof of equality in the lower bound of~\eqref{eq:Ekkcutoff}, which is an open problem due to the difficulties that arise when switching from fixed MAD to fixed variance. 

Figure~\ref{fig:triang25} illustrates these observations: The MAD-extremal random graph contains a number of triangles that grows significantly faster in $n$ than the power-law random graph. The MAD-extremal graph with equal variance as the power-law distribution on the other hand, contains the same number of triangles as the power-law random graph.

 {\it Dense scale-free networks}. 
 In the regime $1<\tau<2$, setting $h_s$ is not as straightforward as in the previous regime. Taking $h_s=h_c=\sqrt{\mu n}$ gives a convex connection probability, but $\mu$ grows in $h_s$ as well. This results in $h_s=h_c\sim n^{1/\tau}$ and $\mu\sim n^{2/\tau-1}$~\cite{bianconi2006}. Under this cutoff~\cite{bianconi2006}, the expected number of cliques in a power-law random graph scales as
 \begin{equation}\label{eq:Kkplsmall}
     \expec_{pl}[K_k]\sim n^{k/\tau}
 \end{equation}
 Now for $1<\tau<2$,~\eqref{eq:pld} yields $d\propto \mu\propto n^{2/\tau-1}$. Then,~\eqref{eq:EKkcl} gives that the maximal number of cliques among all subgraphs with the same MAD as the power-law distribution with exponent $\tau\in(1,2)$ scales as
\begin{align}
   \max_{\mathbb{P} \in \mathcal{P}_{(\mu,d)}}\Expp{N_{K_k}}
	& \sim \frac{n^{k/2}\mu^{-k/2}}{k!2^k}\sim n^{k/\tau},
\end{align}
 which has the same scaling in $n$ as~\eqref{eq:Kkplsmall}. Thus, under the MAD framework, power-law random graphs with $1<\tau<2$ achieve the maximal clique scaling in the network size $n$, and power-law networks are the densest possible networks in terms of cliques.
 
 As in the $2<\tau<3$ regime, $\sigma^2\sim h_s^{3-\tau}$, so that $\sigma^2/h_s\sim h_s^{2-\tau}\sim d$. With $\mu\propto d$ \eqref{eq:EKkvar} yields 
 \begin{align}
      \max_{\mathbb{P} \in \mathcal{P}_{(\mu,\frac{\sigma^2}{h_s-1})}}\Expp{K_k}& \propto \frac{(d h_s)^k}{d^k} \propto n^{k/\tau},
 \end{align}
 which again has the same scaling as the number of cliques in the corresponding power-law random graph obtained from~\eqref{eq:Kkplsmall}.
Thus, this suggests that power-law random graphs contain the maximal amount of triangles among all Chung-Lu models with the same variance, similarly to the $2<\tau<3$ case.

Figure~\ref{fig:triang15} indeed illustrates that the MAD extremal random graph contains a larger number of triangles than a power-law with the same MAD and average degree, but that these numbers scale the same in $n$. The extremal random graph with equal variance as the power-law random graph contains the same number of triangles. 
For $2<\tau<3$ on the other hand, power-laws do not achieve the maximal triangle count scaling under fixed MAD.

Figure~\ref{fig:plK4} summarizes our findings. For $1<\tau<2$, power-law Chung-Lu random graphs achieve the maximal clique scaling among all Chung-Lu random graphs with the same MAD or variance. For $2<\tau<3$, Chung-Lu power-law random graphs achieve a lower number of $k$-cliques than the maximal number of $k$-cliques among all Chung-Lu random graphs with the same MAD. In fact, the maximal number cliques among all Chung-Lu random graphs with the same MAD as power-laws with $2<\tau<3$ scales as the number of cliques in a power-law Chung-Lu random graph for $\tau=2$.

 \begin{figure}[tb]
    \centering
    \includegraphics[width=0.45\textwidth]{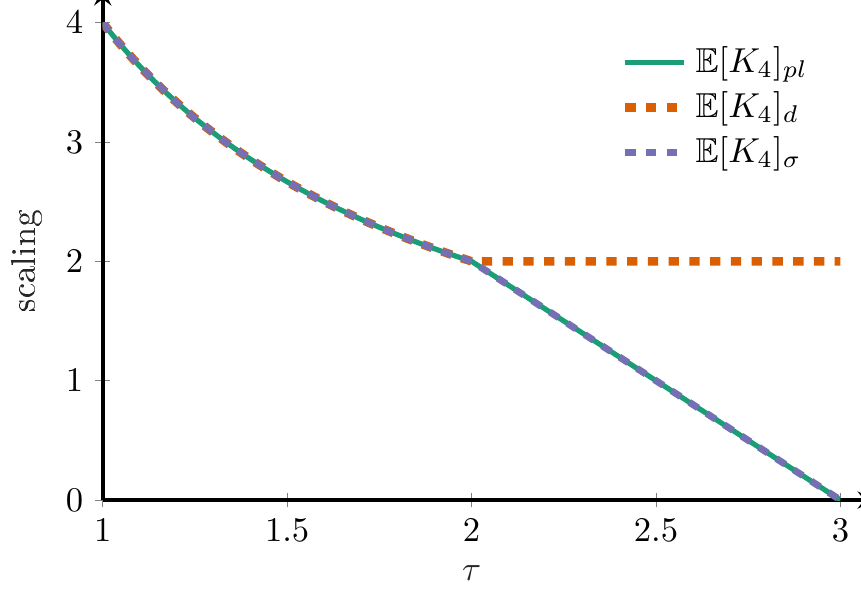}
    \caption{Scaling of the expected number of 4-cliques in $n$. $\mathbb{E}[K_4]_{pl}$ denotes the number of cliques of a power-law random graph with specified degree exponent $\tau$. $\mathbb{E}[K_4]_{d}$ denotes the maximal number of cliques in a Chung-Lu network with the same MAD as a power-law random graph with specified degree exponent $\tau$. $\mathbb{E}[K_4]_{\sigma}$ denotes the maximal number of cliques in a Chung-Lu network with $d=2\sigma^2/(h_s-1)$, where $\sigma^2$ is the variance of a power-law random graph with specified degree exponent $\tau$. }
    \label{fig:plK4}
\end{figure}

Furthermore, Figure~\ref{fig:plK4} suggests that for $1<\tau<3$, power-law random graphs achieve the maximal number of cliques among all weight distributions with the same variance. However, this is only based on a lower-bound technique.

\section{Data}\label{sec:data}
 
 We now investigate the performance of the predicted upper bounds on subgraph counts for nine network data sets that have cutoff below $\sqrt{\mu n}$~\cite{konect}. The data summary statistics are described in Table~\ref{tab:data}. Using only these summary statistics, we can compute the bound on the maximal number of subgraphs in networks with the same $\mu$, $d$ and $n$ using Theorem~\ref{thm:subgraphsmad}. 
 
In Figure~\ref{fig:ratios}, we plot the ratio between the actual subgraph counts in the data sets and this bound on the maximal number of subgraphs for all subgraphs of size four. All subgraphs appear significantly less than the bound predicted by Theorem~\ref{thm:subgraphsmad}. For the collaboration network in network science, the subgraph counts are the closest to the MAD-maximizer, most other subgraphs appear significantly less than the largest MAD bound. Thus, this ratio shows that 4-point subgraphs in the collaboration network of network scientists are closer to maximal than in the other data sets. 

In Figure~\ref{fig:ratiosvar}, we now compare the number of subgraphs with the lower bound on the maximal subgraph count under fixed variance instead of fixed MAD that we obtain from Theorem~\ref{thm:maxNhsigma}. Again, this figure plots the ratio between the actual subgraph counts and the predicted bound from Theorem~\ref{thm:maxNhsigma}. We see that most clustered subgraphs typically appear more often than the predicted bound: the ratios significantly exceed one for several data sets and several subgraphs. One reason for this deviation is that Thoerem~\ref{thm:maxNhsigma} assumes that the largest degree in these networks is equal to $\sqrt{\mu n}$. However, Table~\ref{tab:data} shows that for all data sets, this assumption does not hold. 
We therefore again compare the number of subgraphs with the bounds from Theorem~\ref{thm:subgraphsmad} and~\ref{thm:subgraphsmad}, but now using $h_c=h_{\max}$ instead, where $h_{\max}$ is the maximal degree of the data set. This yields Figure~\ref{fig:ratiosdmax}, which shows that indeed, using the correct cutoff in those networks explains a large part of these large values: now almost all subgraph counts are below the predicted maximal bounds, also for the variance-based lower bound in Figure~\ref{fig:ratiosvardmax}. Still, in three networks, some subgraphs appear more often than the MAD-variance based maximizer. In particular, this happens for the more clustered subgraphs, such as the complete graph on 4 vertices. This therefore shows that indeed, these networks are not generated by hidden-variable models, and are more clustered in terms of their numbers of complete graphs and cycles. 

\begin{figure*}[tb]
    \centering
\subfloat[MAD-based]{
\centering
    \includegraphics[width=0.48\linewidth]{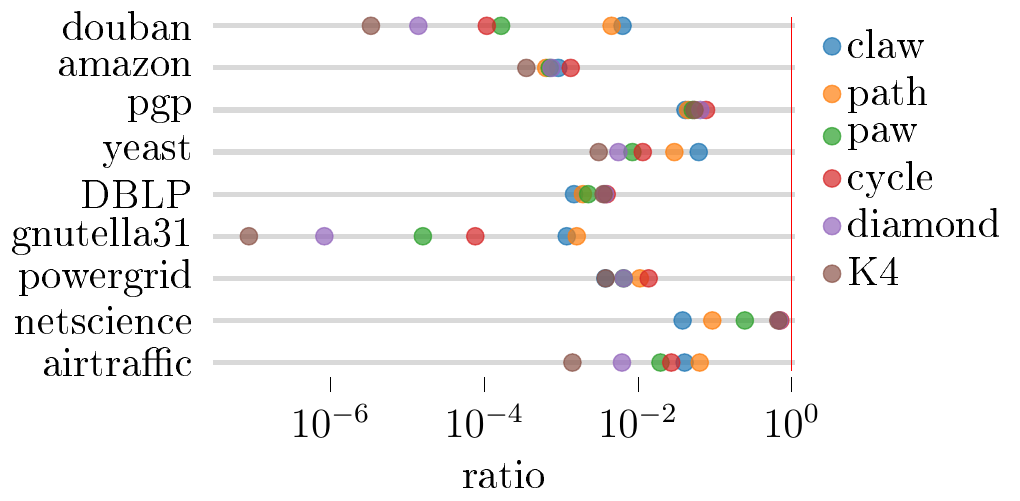}
    \label{fig:ratios}
}
\hfill
\subfloat[MAD-based with $d=\sigma^2/(2(\sqrt{\mu n}-1))$]{
    \includegraphics[width=0.48\linewidth]{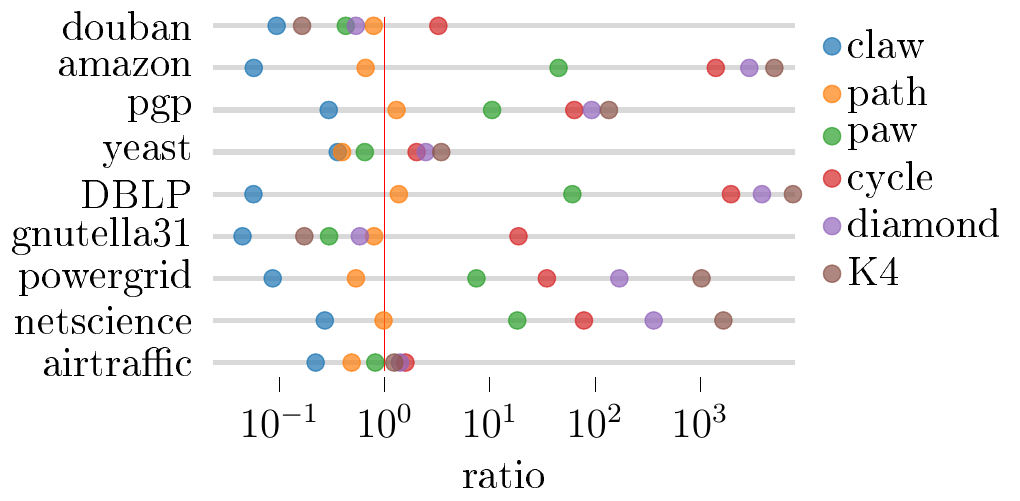}
    \label{fig:ratiosvar}
}
\caption{Ratio of the number of subgraphs of size 4 in nine data sets and the Chung-Lu maximal value of~\eqref{eq:ENHcl} and~\eqref{eq:ENhvar}. The red line indicates the MAD maximal bound.}
\label{fig:ratiosall}
\end{figure*}

\begin{figure*}[tb]
    \centering
\subfloat[MAD-based]{
\centering
    \includegraphics[width=0.48\textwidth]{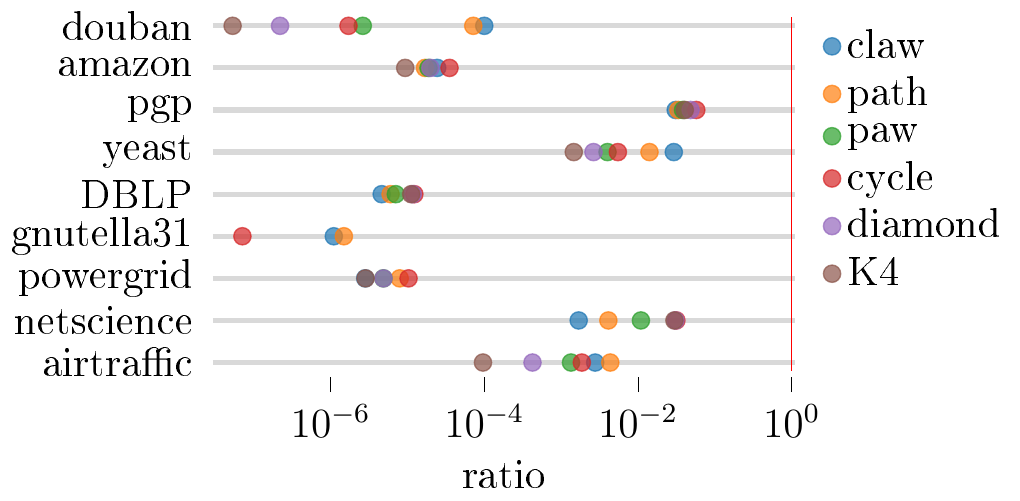}
    \label{fig:ratiosdmax}
}
\hfill
\subfloat[MAD-based with $d=\sigma^2/(2(h_{\max}-1))$]{
    \includegraphics[width=0.48\textwidth]{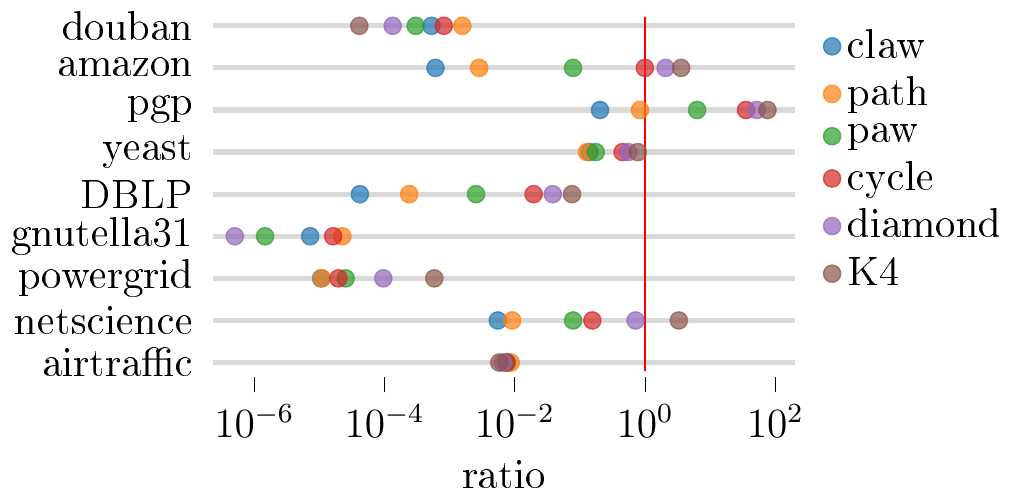}
    \label{fig:ratiosvardmax}
}
\caption{Ratio of the number of subgraphs of size 4 in nine data sets and the Chung-Lu maximal value of~\eqref{eq:ENHcl} and~\eqref{eq:ENhvar}. The red line indicates the MAD maximal bound. In these plots, instead of the default choice of $h_c=\sqrt{\mu n}$, $h_c$ is set to $h_{\max}$, the largest degree of the corresponding data set.}
\label{fig:ratiosalldmax}
\end{figure*}

\begin{table}[tb]
\renewcommand{\arraystretch}{1.2}
\begin{tabular}{  l  l  l  l  l  l  }
\toprule
	Name & $n$ & $\mu$ & $d$ & $h_{\max}$ &$\sigma^2$ \\ \midrule
	Amazon & 334863 & 5.53 & 3.01 & 549 & 33.19 \\ 
	Douban & 154908 & 4.22 & 5.07 & 287 & 138.02 \\ 
	DBLP & 317080 & 6.62 & 5.3 & 343 & 100.15 \\ 
	PGP & 10680 & 4.55 & 4.18 & 205 & 65.24 \\ 
	Yeast & 1870 & 2.44 & 1.72 & 56 & 10.01 \\ 
	Gnutella31 & 62586 & 4.73 & 4.49 & 95 & 32.5 \\ 
	US power grid & 4941 & 2.67 & 1.28 & 19 & 3.21 \\ 
	Netscience & 1461 & 3.75 & 2.28 & 34 & 11.96 \\ 
	airtraffic & 1226 & 4.27 & 2.82 & 37 & 18.72 \\ \bottomrule
\end{tabular}
\caption{Summary statistics of the network data. $h_{\max}$ denotes the maximal network degree.}
\label{tab:data}
\end{table}

\section{Discussion and outlook}\label{sec:out}
We have established distribution-free bounds on subgraph counts, using an optimization method that needs as input only the mean, MAD and range of the degrees. The bounds do not depend on detailed assumptions on a particular network degree distribution, and in fact hold for a wide class of degree distributions. 

The bounds are the sharpest possible and attained by an extremal random graph with a three-point degree distribution. This extremal random graph contains more subgraphs than the popular sparse graphs with power-law degree distributions with $\tau>2$. For dense graphs with $\tau<2$ on the other hand, power-law random graphs match the subgraph bounds. This implies that dense power-law random graphs have the highest possible subgraph counts among all possible degree distributions with the same mean and MAD. 

Furthermore, our bounds indeed bound the subgraph counts of nine real-world data sets, even though these data contain power-law and non power-law degree distributions, and are not generated by hidden-variable models, demonstrating the robustness of our approach.  

We believe that the optimization methid in this paper can be employed for robust analysis of other network statistics as well, such as clustering coefficients and degree correlations. 
Another avenue concerns the relation of the MAD extremal random graph to the maximal possible eigenvalue of graphs with given mean degree and MAD. Cycle counts can be linked to the maximal eigenvalue of the adjacency matrix~\cite{tsourakakis2008}. Therefore, our results on bounds on the maximal number of subgraphs may also provide bounds on the maximal possible eigenvalue of a network with given average degree and MAD. Investigating these eigenvalue bounds further would be an worthwhile topic for further research, especially as the largest eigenvalue of an adjacency matrix is strongly linked to epidemic properties on the network~\cite{pastor2014}. 

While hidden-variable models have proven useful for modeling many types of networks, one of their disadvantages is that they often lead to locally tree-like networks with little clustering. However, this paper shows that when maximizing over all possible hidden-variable distributions, they may contain more clustered subgraphs than real-world networks. Several extensions of hidden-variable models to include clustering exist, including those that add geometry, or higher-order interactions. It would be interesting to apply the MAD framework to these types of models as well. For geometric hidden-variable models, this would lead to difficult optimization problems, as it contains two sources of randomness: the hidden variables and their geometric positions. Investigating if these generally more clustered models also lead to more clustering when maximizing over the variables is an interesting question for further research.

\bibliographystyle{abbrv}
\bibliography{references}

\appendix
\section{Dominating contributions in the extremal random graph}\label{app:scalingMAD}
We now investigate the behavior of~\eqref{eq:NHgen} when $f$ satisfies Assumption~\ref{ass:functionf2}. In that case, when $h_s=h_c$,
\begin{align}\label{eq:ENHsplit2}
	&\Expp{N_H}=\frac{n^k}{\Aut(H)}\nonumber\\
	& \times \sum_{i_1\in\{a,\mu,h_c\}}\!\!\dots\!\! \sum_{i_k\in\{a,\mu,h_c\}}\prod_{j=1}^kp_{i_j}\prod_{\{s,t\}\in E_H}\frac{i_si_t}{h_c^2}r\Big(\frac{i_si_t}{h_c^2}\Big)\nonumber\\
	& =\frac{n^k}{\Aut(H)h_s^{2E_H}}\nonumber\\
	& \times \!\! \sum_{i_1\in\{a,\mu,h_c\}}\!\!\dots \!\!\sum_{i_k\in\{a,\mu,h_c\}}\prod_{j=1}^kp_{i_j}i_j^{d_j}\!\!\prod_{\{s,t\}\in E_H}\!\! r\Big(\frac{i_si_t}{h_c^2}\Big).
\end{align}
We then show that for every vertex $j$, the summation over $i_j\in\{a,\mu,h_c\}$ is dominated by the term containing $h_c$, so that the other terms may be ignored. Writing out this sum over $i_j\in\{a,\mu,h_c\}$ and plugging in~\eqref{eq:3pointp} gives
\begin{align}\label{eq:contrjsplit}
	&\sum_{i_j\in\{a,\mu,h_c\}}p_{i_j}i_j^{d_j}\prod_{\{j,t\}\in E_H}r\Big(\frac{i_ji_t}{h_c^2}\Big) \nonumber\\
	& = \frac{d}{2(\mu -a)}a^{d_j}\prod_{\{j,t\}\in E_H}r\Big(\frac{ai_t}{h_c^2}\Big)\nonumber\\
	& +\Big(1-\frac{d}{2(\mu -a)}-\frac{d}{2(h_c-\mu)}\Big)\mu^{d_j}\prod_{\{j,t\}\in E_H}r\Big(\frac{\mu i_t}{h_c^2}\Big)\nonumber\\
	& \quad +\frac{d}{2(h_c-\mu)}h_c^{d_j}\prod_{\{j,t\}\in E_H}r\Big(\frac{i_t}{h_c}\Big).
\end{align}
This is an equation with three terms. We now investigate the scaling of these three terms in $n$, and show that the last term, containing $h_c$, dominates. Indeed, for the first two terms we obtain
\begin{equation}\label{eq:contra}
	\frac{d}{2(\mu -a)}a^{d_j}\prod_{\{j,t\}\in E_H}r\Big(\frac{ai_t}{h_c^2}\Big)=\bigO{1}
\end{equation}
and 
\begin{equation}
	\Big(1-\frac{d}{2(\mu -a)}-\frac{d}{2(h_c-\mu)}\Big)\mu^{d_j}\!\!\! \prod_{\{j,t\}\in E_H}\!\!r\Big(\frac{\mu i_t}{h_c^2}\Big)=\bigO{1},
\end{equation}
as $r(x)\leq 1$ under Assumption~\ref{ass:functionf2}. We now show that the third term on the other hand, grows in $n$. There we obtain
\begin{equation}\label{eq:contrb}
	\frac{d}{2(h_c-\mu)}h_c^{d_j}\prod_{\{j,t\}\in E_H}r\Big(\frac{i_t}{h_c}\Big)\sim  \frac{dh_c^{d_j-1}}{2}.
\end{equation}
Indeed, as $h_c\to\infty$ and $h_c\gg\mu$, $(h_c-\mu)\sim h_c$. Furthermore, by our assumptions on $r(x)$, $r(c/h_s)=1+o(1)$, while $r(1)$ is constant, so that the product term does not contribute to the scaling.
Thus, when $d_j\geq 2$, the contribution to~\eqref{eq:ENHsplit2} from $i_j=h_c$ grows in $n$, and therefore dominates the contributions from $i_j=a,\mu$. Therefore, for subgraphs with minimal degree at least 2, we can ignore the terms in~\eqref{eq:ENHsplit2} with $i_j=a$ or $i_j=\mu$, yielding 
\begin{align}
	\Expp{N_H}
	& \sim \frac{n^k}{\Aut(H)h_s^{2E_H}}\prod_{j=1}^k\frac{d}{2(h_c-\mu)}h_c^{d_j}\prod_{\{s,t\}\in E_H}r(1)\nonumber\\
	& =\frac{n^kd^kh_c^{2E_H-k}}{h_c^{2E_H}2^k\Aut(H)}r(1)^{E_H}.
\end{align}

When $d_j=1$, all contributions in~\eqref{eq:contra}-\eqref{eq:contrb} have constant order of magnitude in $n$. Therefore, for those vertices, all terms in the summations in~\eqref{eq:ENHsplit2} have to be included, which gives
\begin{align}
	& \Expp{N_H}
	 \sim \frac{n^k}{\Aut(H)h_c^{2E_H}}\prod_{j:d_j\geq 2}\frac{d}{2h_c}h_c^{d_j}\nonumber\\
	 &\quad \times \prod_{j:d_j=1}\left(\frac{dh_c}{2h_c}r(1)+\Big(1-\frac{d}{2(\mu -a)}\Big)\mu +\frac{da}{2(\mu -a)}\right)\nonumber\\
 & \quad \times\prod_{\{s,t\}:d_s,d_t\geq 1\in E_H}r(1)\nonumber\\
	& =\frac{n^kd^{n-n_1}h_c^{2E_H-k}}{\Aut(H)h_c^{2E_H}2^{n-n_1}}\left(\frac{d}{2}(r(1)-1)+\mu\right)^{n_1}r(1)^{E_H-n_1}.
\end{align}
Taken together, we obtain Theorem~\ref{thm:subgraphsmad}.

\section{Dominating contribution for diminishing $d$}\label{app:MADvar}
We start from~\eqref{eq:contrjsplit}, and assume that $\sigma^2/h_c\to 0$. Again, we investigate which of the three terms in the summation over $a,\mu, h_c$ dominates. We therefore investigate the scaling in $n$ of all three terms, and start with the terms containing $a$ and $\mu$.
	Because $r(x)\leq 1$ under Assumption~\ref{ass:functionf2},
	\begin{equation}
		\frac{d}{2(\mu -a)}a^{d_j}\prod_{\{j,t\}\in E_H}r\Big(\frac{ai_t}{h_c^2}\Big)=\bigO{\frac{\sigma^2}{h_c}}
	\end{equation}
	and again
	\begin{equation}
		\Big(1-\frac{d}{2(\mu -a)}-\frac{d}{2(h_c-\mu)}\Big)\mu^{d_j}\prod_{\{j,t\}\in E_H}r\Big(\frac{\mu i_t}{h_c^2}\Big)=\bigO{1}.
	\end{equation}
	We now turn to the term containing $h_c$. Because $h_c\to\infty$ and $r(x)$ is continuous, we have that $r(0)=1$, $r(c/h_s)=1+o(1)$, while $r(1)$ is constant. Therefore,
	\begin{equation}
		\frac{d}{2(h_c-\mu)}h_c^{d_j}\prod_{\{j,t\}\in E_H}r\Big(\frac{i_t}{h_c}\Big)= \Theta\left(\sigma^2h_c^{d_j-2}\right).
	\end{equation}
	This shows that when $d_j\geq 3$, the contribution to~\eqref{eq:ENHsplit} from $i_j=h_c$ dominates the contributions from $i_j=a,\mu$. On the other hand, when $d_j=2$, the contribution from $i_j=h_c$ and from $i_j=\mu$ have the same order of magnitude. Finally, when $d_j=1$,  the contribution to~\eqref{eq:ENHsplit} from $i_j=\mu$ dominates the other contributions. Using that $r(\mu^2/h_c^2)=1+o(1)$ and $r(\mu/h_c)=1+o(1)$, ~\eqref{eq:ENHsplit} becomes
	\begin{align}
		&\Expp{N_H}
		 \sim \frac{n^kr(1)^{E_{\geq 3,\geq 3}}}{\Aut(H)h_c^{2E_H}}\prod_{j:d_j\geq 3}\frac{\sigma^2}{h_c^2}h_c^{d_j}\nonumber\\
		 & \times \prod_{j:d_j= 2}\left(\sigma^2r(1)^{n_{j,\geq 3}}+\mu^2\right)\prod_{j:d_j= 1}\mu \nonumber\\
		& = \frac{n^kr(1)^{E_{\geq 3,\geq 3}}}{\Aut(H)h_c^{2E_H}}\left(\sigma^2r(1)+\mu^2\right)^{n_{2,1}}\left(\sigma^2r(1)^2+\mu^2\right)^{n_2-n_{2,1}}\nonumber\\
		& \times \mu^{n_1}\sigma^{2n_{\geq 3}} h_c^{\sum_{j:d_j\geq 2}(d_j-2)}\nonumber\\
		& = \frac{n^kh_c^{2E_H-2k+n_1}r(1)^{E_{\geq 3,\geq 3}}}{\Aut(H)h_c^{2E_H}}\left(\sigma^2r(1)+\mu^2\right)^{n_{2,1}}\nonumber\\
		& \times \left(\sigma^2r(1)^2+\mu^2\right)^{n_2-n_{2,1}}\mu^{n_1}\sigma^{2n_{\geq 3}} (1+o(1)),
	\end{align}
	where $n_i$ and $n_{\geq i}$ denote the number of vertices of degree $i$ or degree at least $i$ in $H$, and $E_{\geq 3,\geq 3}$ denotes the number of edges between vertices of degree at least 3 in $H$.

\section{Ordering of subgraphs of size 5}\label{sec:motif5order}
\begin{figure*}[tb]
\includegraphics[width=\textwidth]{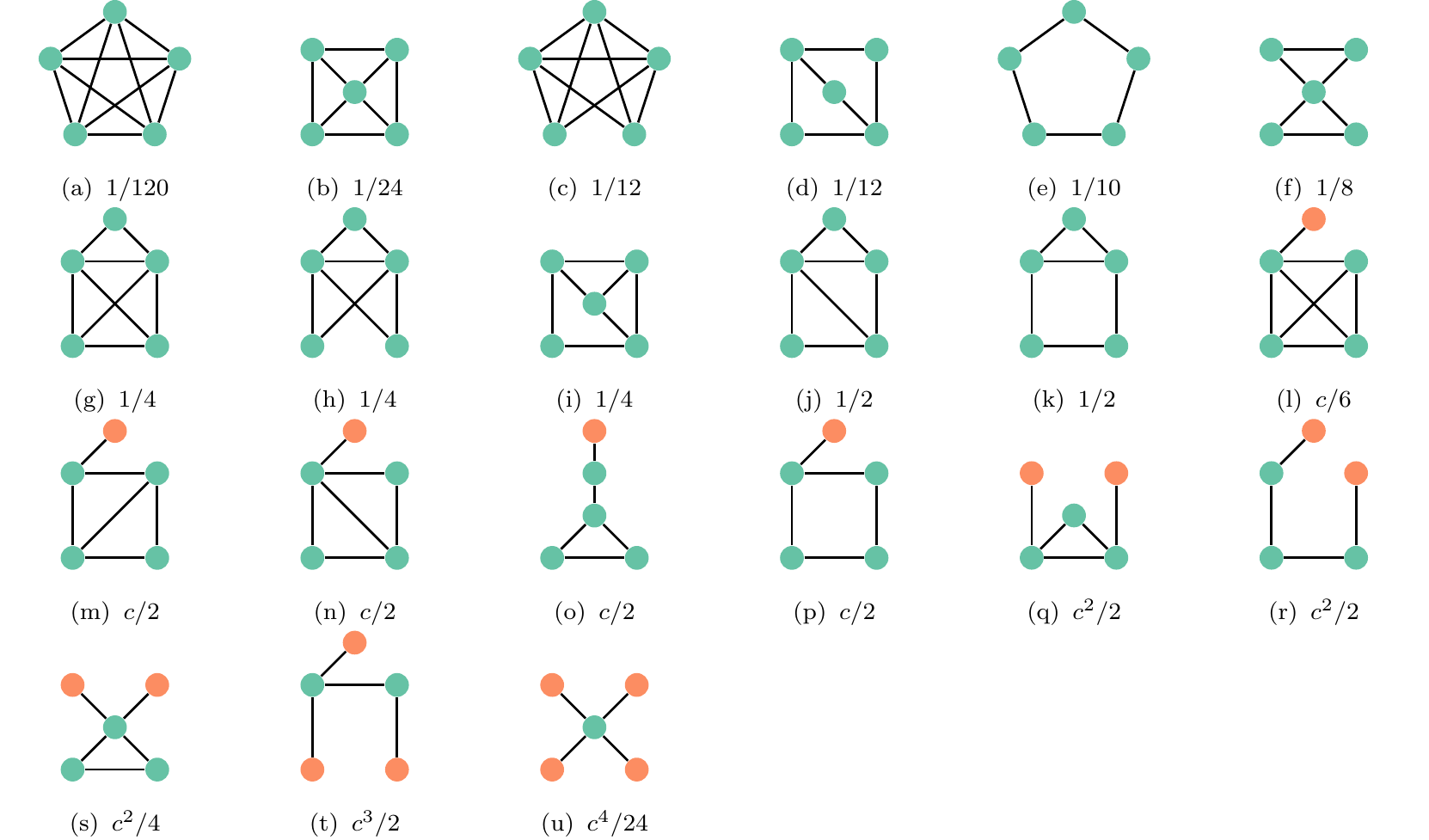}
	\caption{The leading constant $c^{n_1}/\Aut(H)$ for the maximal scaling of the number of subgraphs in $n$ for all subgraphs on 5 vertices, where $c=2\mu/d$.}
	\label{fig:motif5}
\end{figure*}
\end{document}